\documentclass[10pt,conference]{IEEEtran}
\usepackage[margin=0.75in]{geometry}
\usepackage{amsmath}
\usepackage{amsfonts}
\usepackage{graphicx}
\usepackage{natbib}
\usepackage{subfig}
\usepackage[mathscr]{euscript}
\usepackage{amsthm}
\usepackage{enumitem}
\usepackage{blkarray}
\newtheorem{theorem}{Theorem}[section]

\newtheorem{lemma}[theorem]{Lemma}
\theoremstyle{remark}

\theoremstyle{definition}
\newtheorem{definition}{Definition}[section]
\usepackage{url}
\usepackage{arydshln}
\usepackage{mathtools}

\begin{document}
\title{Achieving Secrecy Capacity of Minimum Storage Regenerating Codes for all Feasible $(n,k,d)$ Parameter Values}

\author{\IEEEauthorblockN{V. Arvind Rameshwar}
\IEEEauthorblockA{Department of Electrical Communication Engineering\\
Indian Institute of Science, Bengaluru\\
Email: \texttt{vrameshwar@iisc.ac.in}
}
\and
\IEEEauthorblockN{Navin Kashyap}
\IEEEauthorblockA{Department of Electrical Communication Engineering\\
Indian Institute of Science, Bengaluru\\
Email: \texttt{nkashyap@iisc.ac.in}
}
}
\maketitle
\begin{abstract}
This paper addresses the problem of constructing secure exact-repair regenerating codes at the MSR point for all feasible values of the parameters. The setting involves a passive eavesdropper who is allowed to observe the stored contents of, and the downloads into, an $l$-subset of the $n$ nodes of a distributed storage system (DSS). The objective is to achieve perfect secrecy between the eavesdropped symbols and the file stored on the DSS. Previous secure code constructions (most notably that by Rawat et al.) tackle the problem only for the restricted case wherein the number, $d$, of helper nodes aiding in the recovery of a failed node is equal to $n-1$. This paper builds on Rawat's work, by combining Gabidulin pre-coding and an MSR construction by Ye and Barg to prove the achievability of secrecy capacity at the MSR point for all allowed values of $d$.
\end{abstract}
\IEEEpeerreviewmaketitle
\section{Introduction}

A distributed storage system (DSS) stores a file $\mathbf{f}$ of size $M$ (symbols over a finite field $\mathbb{F}$) on $n$ storage nodes. The system possesses the ``$k$-out-of-$n$" property, in that a data collector (DC) can recover the file by connecting to any $k\text{-}$subset of the nodes. The nodes, however, are prone to failure and the objective is to design schemes that allow for failed-node repair by contacting any $d$ helper nodes, while preserving the ``$k$-out-of-$n$" property. The work by Dimakis et al. \cite{Dimetal} introduced the concept of \textit{regenerating codes}, which address the problem of simultaneous repair and reconstruction while ensuring that each node stores no more than $\alpha$ independent symbols and each helper node passes on no more than $\beta$ independent symbols to the failed node. Then from \cite{Dimetal},
\begin{equation}\label{basic}
    M \leq \sum_{i=1}^{k}\min\{\alpha,(d-i+1)\beta\}.
\end{equation}

The upper bound describes a tradeoff between the parameters $\alpha$ and $\beta$, for a fixed $M$. Two extremal points of this trade-off curve are the minimum storage regeneration (MSR) and the minimum bandwidth regneration (MBR) points. The MSR point, which is of interest to us, is where $\alpha$ is minimized for a given $M$. From \cite{Dimetal} and the tradeoff curve \eqref{basic}, we have
\begin{equation*}
    (\alpha_{MSR},\beta_{MSR}) = \left(\frac{M}{k},\frac{M}{k(d-k+1)}\right).
\end{equation*}

Since MSR codes are equivalent to standard MDS array codes, the goal is to suitably augment MDS array code constructions with repair schemes. MSR codes that meet the capacity upper bound of \eqref{basic} are described in \cite{allparam}, \cite{Rash2}, \cite{vish} and \cite{ye}. In particular, Ye and Barg's constructions in \cite{ye} allow for the parameter $k$ to take on all feasible values (from $1$ to $n$), and similarly, $d$ to take any value in its permissible range of $k+1$ to $n-1$. 

Now consider the \emph{passive eavesdropper} setting, where an eavesdropper, Eve, is allowed to observe, over a long time, the stored contents of, and the downloads into, an $l$-subset of the $n$ nodes. We need to ensure that Eve obtains no information about the file stored in the DSS. 

Capacity upper bounds for \emph{perfect secrecy} at the MSR and MBR points are provided in \cite{Upper}. For the MSR point, work towards tightening the bound in \cite{Upper} can be found in \cite{Gopsecure}, \cite{Huang} and \cite{vish}. While secure codes meeting the capacity upper bound at the MBR point in \cite{Upper} have been constructed for all values of $n,k,d$ \cite{Shah}, the task of constructing secure codes at the MSR point that achieve the improved capacity upper bound in \cite{Gopsecure} has been tackled only for the restricted case of $d=n-1$.

In this work, we provide a secure code construction at the MSR point that achieves the capacity upper bound in \cite{Gopsecure} and \cite{vish} for all values of $n,k,d$, effectively closing the open problem of achieving secrecy capacity at the MSR point. 

In Section \ref{2}, we provide a formal description of the system model and discuss related literature. Section \ref{3} describes the MSR code construction, and provides a proof of secrecy. 
\section{Background and Related Work}\label{2}

In this section, we formally describe the system model, and provide details of Gabidulin-based pre-coding, and an overview of the MSR construction by Ye and Barg \cite{ye}. In what follows, the notation $[a:b]$ denotes the set of integers between $a$ and $b$, both inclusive, i.e., $[a:b] = \{i\in \mathbb{Z}:a\leq i\leq b\}$. We use $[n]$ as shorthand for $[1:n]$.
\subsection{System Model}\label{model}
An $(n,k,d)$ DSS consists of $n$ storage nodes, indexed from $1$ to $n$, that store in a distributed, coded fashion, the $M$ symbols (over a field $\mathbb{F}$) of a file $\mathbf{f}$. The symbols are drawn independently and uniformly at random from the field. 

Let $\mathbf{c}_i$, $i\in [n]$, denote the coded symbols stored in node $i$. Firstly, we require that each node stores no more than $\alpha$ independent symbols. 

We assume that node failures in the system occur in stages, with no more than one failure at any stage. At stage $t$, we say that a node $j$ is \emph{active} if it does not fail in that stage. We operate in the \emph{exact-repair} setting, wherein the downloads from $d$ active helper nodes ($k+1\leq d\leq n-1$) can exactly recover the contents of the failed node. In keeping with \cite{Dimetal}, our second constraint is that the failed node downloads no more than $\beta$ independent symbols from any one helper node. 

Now, suppose that node $i$ has failed. Let $D_{j,i}$ denote the collection of random symbols sent by helper node $j$ to $i$. If $H(X)$ represents the entropy of a random variable $X$, then
\begin{align}
    H(\mathbf{c}_i) &\leq \alpha, \label{e1}\\
    H(D_{j,i}) &\leq \beta. \label{e2}
\end{align}
From \cite{interior}, we know that exact-repair codes that satisfy \eqref{basic} with equality must also satisfy \eqref{e1} and \eqref{e2} with equality. 

Since at the MSR point, $\alpha$ is minimized for a given $M$, we have from \eqref{basic} that $\alpha = M/k$. The minimum value of $\beta$ then is $\beta = \alpha/(d-k+1)$. 
Hence,
\begin{equation*}
    (\alpha,\beta) = \left(\frac{M}{k},\frac{M}{k(d-k+1)}\right).
\end{equation*}

Now consider the case where an eavesdropper, Eve, observes the downloaded symbols into an arbitrary $l\text{-}$subset $\mathscr{E}$ of the nodes. We assume that each node in $\mathscr{E}$ may fail multiple times, and in order to repair the same node over and over again, information is possibly downloaded from different sets of helper nodes. Thus, over time, Eve knows the stored contents of the nodes in $\mathscr{E}$, and has observed repair information for nodes in $\mathscr{E}$ from all nodes not in $\mathscr{E}$. Let the random vector $\mathbf{e}$ denote the symbols observed by Eve. Thus, $\mathbf{e}$ consists of $\mathbf{c}_i$, $i \in \mathscr{E}$, as well as all the $D_{j,i}$, $i \in \mathscr{E}$, $j \notin \mathscr{E}$. If $\mathbf{f}^{(s)}$ ($(s)$ for ``secure") is the file that we desire to store on the DSS, and $M^{(s)}$ is its size, the \emph{perfect secrecy} condition then is: $I(\mathbf{f^{(s)}}; \mathbf{e}) = 0$, where $I(\mathbf{x};\mathbf{y})$ is the mutual information between the random vectors $\mathbf{x}$ and $\mathbf{y}$.

\subsection{Related Work}
The setting of the passive adversary was first discussed in \cite{Upper}, and an upper bound on the secrecy capacity for functional repair was derived to be
\begin{equation*}
    M^{(s)}\leq \sum_{i=l+1}^{k}\min\{\alpha,(d-i+1)\beta\},
\end{equation*}
where $l = |\mathscr{E}|$. Later work by Shah et al. \cite{Shah} employed the Product-Matrix (PM) code construction to design a secure MSR coding scheme that achieved a maximum file size of $(k-l)(\alpha-l\beta)$. This was improved upon in \cite{vish} and \cite{Gopsecure}, wherein the secrecy capacity was shown to be bounded as
\begin{equation}\label{cap}
    M^{(s)}\leq (k-l)(1-1/(d-k+1))^{l}\alpha.
\end{equation}
This upper bound was shown to be achievable in \cite{vish}, for the case $n=d+1$, using the concept of zigzag codes. Another achievability scheme, due to Rawat \cite{Rawat}, uses a construction in Ye and Barg's paper \cite{ye} to show the capacity upper bound in \eqref{cap} being met, again when $n=d+1$. In this paper, we build upon Rawat's work to prove the achievability of the capacity upper bound in \eqref{cap} for \emph{all} feasible values of $d$, using an alternative construction from \cite{ye}.

The tools we need are provided by the work of Huang et al.\ \cite{Huang}. Recall that in a DSS, a given helper node $j$ may in general belong to multiple repair groups (sets of helper nodes) for a given failed node $i$. A distributed storage code operating at the MSR point is said to be \emph{stable} \cite[Definition~7]{Huang} if for each pair of nodes $i$ and $j$, the information downloaded from node $j$ to repair node $i$ is the same across all repair groups for $i$ containing the node $j$. In Lemma~7 of \cite{Huang}, it is shown that for DSSs based on a stable MSR code, the secrecy capacity of an $(n=d+1,k,d)$ DSS is the same as that of an $(n>d+1,k,d)$ DSS, when all other parameters are identical. This result, along with the observations of Rawat \cite{Rawat}, in fact suffices to establish that the construction we describe in Section~\ref{3} achieves the secrecy capacity upper bound in \eqref{cap}. We, however, give a more direct proof, involving some ideas from hypergraph theory that may be of independent interest.

\subsection{Preliminaries}

Given a DSS that can store $M$ symbols when $l=0$, we augment our file of size $M^{(s)}$ with random symbols $\mathbf{r}$, where $\mathbf{r}$ is a random vector of length $R = M-M^{(s)}$. Each random symbol in $\mathbf{r}$ is drawn i.i.d. and uniformly at random from the field $\mathbb{F}$. We shall now describe the ingredients of our construction, namely the Gabidulin pre-coding procedure and the \textit{d-optimal repair} MSR construction (for all parameters $n,k,d$), by Ye and Barg \cite{ye}.\\

\subsubsection{Gabidulin Pre-coding}\label{gab}
Assume that we have an $M$-length vector $\mathbf{m} = (m_1,\ldots,m_M)$, where each $m_i$, $1\leq i\leq M$, is drawn from a finite field $\mathbb{F}$. Let $\mathbb{B}$ be some sub-field of $\mathbb{F}$. Further, let points $y_1,y_2,\ldots,y_M$, be elements of $\mathbb{F}$ that are linearly independent over $\mathbb{B}$ ($\text{dim}_{\mathbb{B}}(\mathbb{F}) \geq M$).

The procedure for Gabidulin coding is:

\begin{itemize}
    \item First, a linearized polynomial $p_{\mathbf{m}}(x)$ is constructed:
    \begin{equation*}
        p_{\mathbf{m}}(x) = \sum_{i=0}^{M-1}m_{i+1}x^{|\mathbb{B}|^{i}}
    \end{equation*}
 \item The polynomial is evaluated at the collection $\mathscr{Y}$ of points $y_1,y_2,\ldots,y_M$, yielding
\begin{equation*}
    p(\mathbf{m},\mathscr{Y}) := (p_m(y_1),p_m(y_2),\ldots,p_m(y_M)).
\end{equation*}
\end{itemize}

\subsubsection{Ye and Barg construction}\label{yeb}
Here, we provide a brief description of the \textit{d-optimal repair} construction.

\textbf{Construction 1:} First we shall introduce some notation: let $s=d-k+1$ and let $\alpha = s^{n}$. Let $\mathbb{F}$ be a finite field of size $|\mathbb{F}|\geq sn$ and let $\{e_a:a\in[0:\alpha-1]\}$ be the standard basis of $\mathbb{F}^{\alpha}$ over $\mathbb{F}$. 

For an integer $a\in [0:\alpha-1]$, let $\mathbf{a} = (a_{n},a_{n-1},\ldots,a_1)$ denote its $s-$ary representation so that $\textbf{a} = \sum_{i=1}^n a_i s^{i-1}$. Suppose that $\{\lambda_{i,j}\}$ for $i\in[n]$ and $j\in [0:s-1]$ are $sn$ distinct elements in $\mathbb{F}$. We define the matrices $A_i$, $i\in [n]$, to be $\alpha \times \alpha$ diagonal matrices with the $(a,a)^{th}$ entry being $\lambda_{i,a_i}$. In other words,
\begin{equation}
    A_i = \sum_{a=0}^{\alpha-1}\lambda_{i,a_i}e_ae_a^{T}.
\end{equation}

We shall construct an $(n-k)\alpha \times n\alpha$ parity check matrix $H$ for the MSR code $\mathscr{C}$ as:
\begin{equation}
    H = \begin{pmatrix}
    I & \ldots & I & I \\
    A_1 & \ldots & A_{n-1} & A_n \\
    A_1^{2} & \ldots & A_{n-1}^{2} & A_n^{2} \\
    \vdots & \ddots & \vdots & \vdots \\
    A_1^{n-k-1} & \ldots & A_{n-1}^{n-k-1} & A_n^{n-k-1}
    \end{pmatrix}
\end{equation}
where $I$ is the $\alpha \times \alpha$ identity matrix.

In \cite{ye}, the authors prove that the code $\mathscr{C}$ obeys the ``$k$-out-of-$n$" property, while storing exactly $\alpha$ independent symbols in each node. In addition, it is proved that the exact-repair requirement of the DSS is also met, ensuring that the contents of any one failed node can be exactly recovered from $\beta = \frac{\alpha}{d-k+1} = \frac{s^n}{s} = s^{n-1}$ symbols from each of $d$ other active nodes.

We shall now describe the repair scheme. Let $\mathbf{c} = \{\mathbf{c}_1,\ldots,\mathbf{c}_n\}$ be a codeword of $\mathscr{C}$. Since $\alpha = s^{n}$, we shall index the symbols in $\mathbf{c}_i$ by $n$-tuples from $[0:s-1]^{n}$. Let $\mathscr{S}=[0:s-1]$ and $\mathscr{S}^{n}=[0:s-1]^{n}$. The symbols in $\mathbf{c}_i$ are indexed by the vectors $\mathbf{a}\in \mathscr{S}^{n}$, in lexicographic order, starting with the vector $\mathbf{0}$ and ending with the vector $\mathbf{z}$ (the $s$-ary representation of $\alpha -1$). In vectorized form, $\mathbf{c}_i$ is the $\alpha\times 1$ column vector given as
\begin{equation*}
    \mathbf{c}_i = (c_{i,\mathbf{0}},\ldots,c_{i,\mathbf{z}})^{T}.
\end{equation*}

For a vector $\mathbf{a}\in \mathscr{S}^{n}$, let $(a_n,a_{n-1},\ldots,a_1)$ be its $s$-ary representation. Now, let $\mathbf{a}(i,u)\in \mathscr{S}^{n}$ be the vector obtained by substituting the symbol $a_i$ in the $s$-ary representation of $\mathbf{a}$, with $u$, for $i\in [n]$ and $u\in [0:s-1]$. Thus,
\begin{equation*}
    \mathbf{a}(i,u) \equiv (a_n,a_{n-1},\ldots,a_{i+1},u,a_{i-1},\ldots,a_1).
\end{equation*}

Assume that node $i$ has failed and hence, $\mathbf{c}_i$ needs to be recovered. Recall that each helper node $j$, $j\neq i$, sends exactly $\beta = s^{n-1}$ independent symbols to node $i$. For some $\mathbf{a} \in \mathscr{S}^{n}$, we define the set $\mathscr{S}^{n}_{\mathbf{a},i}$ as
\begin{equation*}
    \mathscr{S}^{n}_{\mathbf{a},i} := \{\mathbf{a}(i,u):\text{ }u\in \mathscr{S}\}.
\end{equation*}
Note that $|\mathscr{S}^{n}_{\mathbf{a},i}|=s$, for any $\mathbf{a}\in \mathscr{S}^{n}$. Furthermore,
\begin{equation*}
    \bigcup_{\mathbf{a}}\mathscr{S}^{n}_{\mathbf{a},i} = \mathscr{S}^{n}.
\end{equation*}
Thus, there exist $\beta = s^{n-1}$ distinct sets $\mathscr{S}^{n}_{\mathbf{a},i}$, the union of which is the entire set of $s$-ary $n$-tuples. We shall use these $\beta$ distinct sets to index the symbols sent by node $j$ to failed node $i$. 

Now, let $D_{j,i}$ represent the symbols contributed by helper node $j$ towards the repair of node $i$ ($j\neq i$). Thus, from \cite{ye}, $D_{j,i}$ is the row vector of the $\beta$ symbols
\begin{equation}\label{down}
    \mu_{j,i}^{(\mathscr{S}^{n}_{\mathbf{a},i})} = \sum_{u=0}^{s-1}c_{j,\mathbf{a}(i,u)}
\end{equation}
for \emph{distinct} sets $\mathscr{S}^{n}_{\mathbf{a},i}$, $\mathbf{a} \in \mathscr{S}^{n}$. Observe that $\mathscr{C}$ is a stable MSR code, in the sense of \cite[Definition~7]{Huang}. 

\section{Secure MSR codes for all parameters}\label{3}
In this section, we describe our construction of secure MSR codes for all feasible values of $d$, using arguments from \cite{Rawat}.

\textbf{Construction 2:} Consider a file $\mathbf{f^{(s)}}$, that we intend storing on the DSS, of size
\begin{equation}\label{size}
    M^{(s)} = (k-l)(1-1/(d-k+1))^{l}\alpha
\end{equation}
symbols, over a field $\mathbb{F}$. The file size in \eqref{size} meets the secrecy capacity upper bound at the MSR point, derived in \cite{Gopsecure}. As in the Ye and Barg construction in Section \ref{yeb}, we take $\alpha = s^{n}$, where $s=d-k+1$, for $k+1\leq d\leq n-1$. We now describe our coding scheme: 

\begin{enumerate}
    \item \textbf{Gabidulin pre-coding:} To the information set of size $M^{(s)}$, we add $R = M-M^{(s)}$ random symbols (denoted by the vector $\mathbf{r}$), drawn i.i.d. and uniformly from the field $\mathbb{F}$, where $M = k\alpha$. Let this overall message $\mathbf{m} = (\mathbf{f^{(s)}},\mathbf{r})$ be Gabidulin pre-coded by the procedure described in \ref{gab}. Let
    \begin{equation*}
        \mathbf{f} :=  p(\mathbf{m},\mathscr{Y}) = (p_m(y_1),p_m(y_2),\ldots,p_m(y_M)).
            \end{equation*}
    \item \textbf{Ye and Barg encoding:} Let $H$ be the parity check matrix of the Ye and Barg code specified in Construction 1 of Section \ref{yeb}. The $k\alpha \times n\alpha$ generator matrix of the code, $G$, satisfies
    \begin{equation*}
        GH^{T} = \mathbf{0},
    \end{equation*}
    where $\mathbf{0}$ denotes the $k\alpha \times (n-k)\alpha$ zero matrix. The code vector $\mathbf{c} = (\mathbf{c}_1,\ldots,\mathbf{c}_n)$ to be stored in the nodes of the DSS is obtained as
    \begin{equation*}
        \mathbf{c} = \mathbf{f}G \in \mathscr{C},
    \end{equation*}
    where $\mathbf{f}$ is the pre-coded vector from Step 1. The $i^{th}$ node of the DSS stores the vector $\mathbf{c}_i$ of $\alpha$ symbols.
\end{enumerate}
From the discussion in Section \ref{yeb} and from \cite{ye}, we know that the coding scheme described above is MSR and satisfies the exact-repair property for all values of $d$. The proof of secrecy follows next.

\subsection{Proof of secrecy for $k+1\leq d\leq n-1$}
We shall follow the line of argument presented in \cite{Rawat}. Let the set of nodes that Eve eavesdrops on be $\mathscr{E} = \{i_1,i_2,\ldots,i_l\}$. In the worst case, all nodes in $\mathscr{E}$ have failed at least once. Note that, as before, we require $|\mathscr{E}|=l<k$. Further, let $D_{j,\mathscr{E}}$ represent the symbols sent by the $j^{th}$ active node ($j \in \mathscr{D} \subset [n]\setminus \mathscr{E}$, such that $|\mathscr{D}|=d$), for the repair of the nodes in $\mathscr{E}$. Hence,
\begin{equation*}
    D_{j,\mathscr{E}} = [D_{j,i_1} \big \rvert D_{j,i_2} \big \rvert \ldots \big \rvert D_{j,i_l}],
\end{equation*}
where the solid vertical lines represent concatenation.

Without loss of generality, we assume that $\mathscr{E} = \{n-l+1,n-l+2,\ldots,n\}$, for, if otherwise, we can always reorder the nodes prior to the first node failure. To characterize the symbols downloaded by the nodes in $\mathscr{E}$, we make the following definition.

\begin{definition}(Symbol Matrix):
A symbol matrix $P$ corresponding to the repair scheme \eqref{down} is a $0\text{-}1$ matrix of dimension $l\beta \times \alpha$ such that $D_{j,\mathscr{E}}^{T}=P\mathbf{c}_j$ for all $j \in [n]\setminus \mathscr{E}$. 
\end{definition}

In order to explicitly describe the entries of $P$, we require some notation. Recall from Section \ref{2} that $\mathscr{S}^{n}$ represents the set of vectors in $[0:s-1]^{n}$, and that $\mathscr{S}^{n}_{\mathbf{a},i}=\{\mathbf{a}(i,u):\text{ }u\in \mathscr{S}\}$, for $\mathbf{a}\in \mathscr{S}^{n}$. Now, define
\begin{equation*}
    \mathscr{S}_{i\leftarrow *}^{n} := \{(a_n,\ldots,a_1):\text{ } a_i=*,\text{ }a_j\in \mathscr{S}\text{ for }j\neq i\}.
\end{equation*}
Note that for any $i$, $|\mathscr{S}_{i\leftarrow *}^{n}| = s^{n-1}$.

For a vector $\mathbf{a} \in \mathscr{S}^{n}$ (or $\mathscr{S}^{n}_{j\leftarrow *}$ for some $j$), let $\textbf{a}_{\setminus i}$ denote the vector obtained by puncturing $\mathbf{a}$ in its $i^{th}$ coordinate:
\begin{equation*}
    \textbf{a}_{\setminus i} = (a_n,\ldots,a_{i+1},a_{i-1},\ldots,a_1).
\end{equation*}
Now, from the definition of the symbol matrix $P$, we have
\begin{equation}\label{P}
    P = \left[
\begin{array}{ccc}
  \hfill & P_n & \hfill\\
  \hline
  \hfill & P_{n-1} & \hfill\\
  \hline
  \hfill & \vdots & \hfill\\
  \hline
  \hfill & P_{n-l+1} & \hfill
\end{array}
\right]
\end{equation}
where each $P_i$, $i\in [n-l+1:n]$ is a 0-1 matrix of dimensions $\beta\times \alpha$, such that $P_i\mathbf{c}_j = D_{j,i}^{T} = \left(\mu_{j,i}^{(\mathscr{S}^{n}_{\mathbf{0},i})},\ldots,\mu_{j,i}^{(\mathscr{S}^{n}_{\mathbf{z},i})}\right)^{T}$, with $\mu_{j,i}^{(\mathscr{S}^{n}_{\mathbf{a},i})}$ as in \eqref{down}.

We now seek to characterize $P_i$, $i\in [n-l+1:n]$, completely. Let the columns of $P_i$ ($i\in [n-l+1:n]$) be indexed by all the vectors in $\mathscr{S}^{n}$, listed in lexicographic order and let the rows of $P_i$ be indexed by the vectors $\mathbf{b}\in \mathscr{S}_{i\leftarrow *}^{n}$, in lexicographic order. 

From \eqref{down}, we see that the row in $P_i$ indexed by some $\mathbf{b}\in \mathscr{S}_{i\leftarrow *}^{n}$ contains exactly $s$ 1's --- these are in the columns indexed by the vectors $\mathbf{b}(i,u) = (b_n,\ldots,b_{i+1},u,b_{i-1},\ldots,b_1)$, for $u\in [0:s-1]$. All other entries of $P_i$ are 0's. Note that the column indices containing a 1 entry differ in exactly their $i^{th}$ coordinate. Explicitly,
\begin{equation}\label{P_1}
    {[P_i]}_{\mathbf{r},\mathbf{t}} = \begin{cases}
    1,\text{ if }\textbf{t}_{\setminus i} = \textbf{r}_{\setminus i},\\
    0,\text{ otherwise}
    \end{cases}
\end{equation}
for $\mathbf{r}\in \mathscr{S}_{i\leftarrow *}^{n}$ and $\mathbf{t}\in \mathscr{S}^{n}$.

Further, equation \eqref{P} coupled with equation \eqref{P_1} above, implies that each column of $P$ contains exactly $l$ 1's, one in each $P_i$.

Equations \eqref{P} and \eqref{P_1} completely characterize $P$. We add that $H(D_{j,\mathscr{E}}) = \text{rank}(P)$, since the symbols in $\mathbf{c}_j$ are independent of one another.

As an example, consider the $(n,k,d,l) = (4,2,3,1)$ DSS wherein Eve eavesdrops on the last (fourth) node. The symbol matrix $P$ in this case is:
\begin{equation*}
\begingroup 
\setlength\arraycolsep{2pt}
    P = \begin{bmatrix}
    1 & 0 & 0 & 0 & 0 & 0 & 0 & 0 & 1 & 0 & 0 & 0 & 0 & 0 & 0 & 0\\
    0 & 1 & 0 & 0 & 0 & 0 & 0 & 0 & 0 & 1 & 0 & 0 & 0 & 0 & 0 & 0\\
    0 & 0 & 1 & 0 & 0 & 0 & 0 & 0 & 0 & 0 & 1 & 0 & 0 & 0 & 0 & 0\\
    0 & 0 & 0 & 1 & 0 & 0 & 0 & 0 & 0 & 0 & 0 & 1 & 0 & 0 & 0 & 0\\
    0 & 0 & 0 & 0 & 1 & 0 & 0 & 0 & 0 & 0 & 0 & 0 & 1 & 0 & 0 & 0\\
    0 & 0 & 0 & 0 & 0 & 1 & 0 & 0 & 0 & 0 & 0 & 0 & 0 & 1 & 0 & 0\\
    0 & 0 & 0 & 0 & 0 & 0 & 1 & 0 & 0 & 0 & 0 & 0 & 0 & 0 & 1 & 0\\
    0 & 0 & 0 & 0 & 0 & 0 & 0 & 1 & 0 & 0 & 0 & 0 & 0 & 0 & 0 & 1\\
    \end{bmatrix}.
\endgroup
\end{equation*}
It is easy to verify that rank$(P)$ above is 8, which in turn is equal to $H(D_{j,\mathscr{E}})$.

We intend to obtain a handle on the rank of $P$, in general. To this end, we claim that the following theorem holds true:
\begin{theorem}\label{s}
For $s=d-k+1\geq 2$,
\begin{equation}
    H(D_{j,\mathscr{E}}) = s^{n-l}(s^{l}-(s-1)^{l}).
\end{equation}
In other words, the rank of the symbol matrix $P$ is $s^{n-l}(s^{l}-(s-1)^{l})$.
\end{theorem}
We shall defer the proof of Theorem \ref{s} until later, and prove the following theorem, based on the validity of Theorem \ref{s}.

\begin{theorem}\label{sec}
The coding scheme of Construction 2 is secure, for $k+1\leq d\leq n-1$, against a passive eavesdropper that has access to a set $\mathscr{E}\subset [n]$ of nodes, with $|\mathscr{E}|=l$.
\end{theorem}
\begin{proof}
The proof of the theorem is similar to the proof of Proposition 1 in \cite{Rawat}. We intend showing that $H(\mathbf{e})\leq H(\mathbf{r})$ and $H(\mathbf{r}|\mathbf{e},\mathbf{f^{(s)}}) = 0$, thereby implying (from the perfect secrecy lemma of \cite{Shah}) that $I(\mathbf{f^{(s)}};\mathbf{e})=0$. Let $\mathscr{T}$ represent a group of $k-l$ nodes such that $\mathscr{T}\cap \mathscr{E} = \emptyset$. We know that 
\begin{equation*}
    \mathbf{e} = (\mathbf{c}_i:i \in \mathscr{E}) \cup \biggl( \bigcup_{i\in \mathscr{E}} \bigcup_{j \in [n] \setminus \mathscr{E}} \{D_{j,i}\}\biggr).
\end{equation*}
Now, using the notation $\mathbf{c}_{\mathscr{E}} := (\mathbf{c}_i:i \in \mathscr{E})$, we have
\begin{align}
    H(\mathbf{e}) &= l \alpha + H\biggl(\bigcup_{i\in \mathscr{E}} \bigcup_{j \in [n] \setminus \mathscr{E}} \{D_{j,i}\} \ \bigg | \  \mathbf{c}_{\mathscr{E}} \biggr) \notag \\ 
    &= l\alpha + H\biggl(\bigcup_{i\in \mathscr{E}} \bigcup_{j \in \mathscr{T}} \{D_{j,i}\} \ \bigg | \  \mathbf{c}_{\mathscr{E}}  \biggr)  \label{second} \\
    &\leq l\alpha + H\biggl(\bigcup_{i\in \mathscr{E}} \bigcup_{j \in \mathscr{T}} \{D_{j,i}\} \biggr) \notag \\
    &\leq l\alpha + \sum_{j\in \mathscr{T}}H(D_{j,\mathscr{E}}) \notag\\
    &= ls^{n} + (k-l)(1-(1-1/s)^{l})s^{n} \notag\\
    &= ks^{n} - (k-l)(1-1/s)^{l}s^{n} \notag\\
    &= M - M^{(s)} = H(\mathbf{r}). \notag
\end{align}
The equality in \eqref{second} follows from the fact that
\begin{align}
    H & \biggl(\bigcup_{i\in \mathscr{E}} \bigcup_{j \in [n] \setminus (\mathscr{T} \cup \mathscr{E})} \{D_{j,i}\} \ \bigg | \  \mathbf{c}_{\mathscr{E}},  \bigcup_{i\in \mathscr{E}} \bigcup_{j \in \mathscr{T}} \{D_{j,i}\}\biggr) \notag \\
    & \le \sum_{i \in \mathscr{E}} H\biggl(\bigcup_{j \in [n] \setminus (\mathscr{T} \cup \mathscr{E})} \{D_{j,i}\} \ \bigg | \  \mathbf{c}_{\mathscr{E}},  \bigcup_{j \in \mathscr{T}} \{D_{j,i}\}\biggr) \label{third} \\
    & = 0, \notag
\end{align}
the last equality holding since each summand in \eqref{third} equals $0$ (from the proof of Lemma~7 in \cite{Huang}).

Using the MDS array property of the Ye and Barg code and from Remark 8 of \cite{vish}, it is possible to show that $H(\mathbf{r}|\mathbf{e},\mathbf{f^{(s)}}) = 0$. We refer the reader to the proof of Proposition 1 in \cite{Rawat}, for more details.

Now, from the perfect secrecy lemma in \cite{Shah}, the two conditions above imply that $I(\mathbf{f^{(s)}};\mathbf{e})=0$, thereby proving that perfect secrecy holds.
\end{proof}

We shall now proceed to the proof of Theorem \ref{s}, beginning with the definitions of a few notions related to hypergraphs.

\begin{definition}(Incidence matrix)
The incidence matrix (or vertex-edge incidence matrix) $V$ of a hypergraph $(X,E)$ is a 0-1 matrix of dimension $|V|\times |E|$, with the rows representing nodes and columns representing hyperedges, such that $V_{i,j} = 1$ if edge $j$ is incident on vertex $i$, and $0$ otherwise.
\end{definition}

For a vector $\mathbf{v}$, we define its support to be the set of coordinates in which $\mathbf{v}$ takes on non-zero values. 

\begin{definition}(Connected hypergraph)
A hypergraph $(X,E)$ is said to be connected, if for every pair of nodes $(u,w)\in X\times X$, $u\neq w$, there exists an alternating sequence of nodes and hyperedges, $v_0,h_0,v_1,h_1,\ldots,v_{m-1},h_{m-1},v_m$, ($m\in \mathbb{Z}_+$) with $v_0=u$ and $v_m=w$, such that for $i\in [0:m-1]$, $h_i$ is incident on both $v_{i}$ and $v_{i+1}$. We call the sequence of hyperedges $h_0,h_1,\ldots,h_{m-1}$ as a \emph{path} from $u$ to $w$.
\end{definition}

Now, we denote by $\mathscr{G}_{s,n}$, the $n$-dimensional regular hypergraph $(X,E)$ with $|X|=s^{n}$ and $E\subset X^{s}$, with incidence matrix $V_{\mathscr{G}_{s,n}}$ defined as follows: let the rows of $V_{\mathscr{G}_{s,n}}$ be indexed by all the vectors in $\mathscr{S}^{n}$, listed in lexicographic order. Further, let the columns of $V_{\mathscr{G}_{s,n}}$ be indexed by the vectors $\mathbf{b}\in \mathscr{S}_{i\leftarrow *}^{n}$, ($i$ ranging from $n$ down to 1), where for any $i$, the vectors $\mathbf{b}$ are listed in lexicographic fashion. Hence, the first $s^{n-1}$ columns of $V_{\mathscr{G}_{s,n}}$ are indexed in lexicographic order by vectors in $\mathscr{S}_{n\leftarrow *}^{n}$, the next $s^{n-1}$ columns are indexed by vectors in $\mathscr{S}_{(n-1)\leftarrow *}^{n}$, and so on. Thus, there are $ns^{n-1}$ columns overall. The entries of $V_{\mathscr{G}_{s,n}}$ are

\begin{equation}\label{V}
    {[V_{\mathscr{G}_{s,n}}]}_{\mathbf{r},\mathbf{t}} = \begin{cases}
    1,\text{ if }\textbf{r}_{\setminus i} = \textbf{t}_{\setminus i}\text{ and }t_i=*\\
    0,\text{ otherwise}
    \end{cases}
\end{equation}
where $\mathbf{r}\in \mathscr{S}^{n}$ and $\mathbf{t}\in \mathscr{S}_{i\leftarrow *}^{n}$, $i\in [n]$. 

The column in $V_{\mathscr{G}_{s,n}}$ indexed by the vector $\mathbf{b} \in \mathscr{S}_{i\leftarrow *}^{n}$ for some $i$, has exactly $s$ 1's in precisely those rows $\mathbf{t}$ for which $t_i=u$, $u\in [0:s-1]$ and $t_j=b_j$, for $j\neq i$. Moreover, each row of $V_{\mathscr{G}_{s,n}}$ has exactly $n$ 1's.

With the aid of the definitions above, we wish to prove Theorem \ref{s}. Two lemmas follow. The first (Lemma~\ref{extra}) establishes that the transpose of the symbol matrix $P$ can be thought of as the incidence matrix of a regular hypergraph having exactly $s^{n-l}$ connected components. Our second lemma (Lemma~\ref{kro}) proves that the rank of the incidence matrix corresponding to each of these connected components is $(s^{l}-(s-1)^{l})$. We conclude by showing that the rank of $P$ is precisely the sum of the ranks of the incidence matrices of these connected components.

\begin{lemma}\label{extra}
$P^{T}$ is the incidence matrix of a subgraph $\mathscr{H}$ of $\mathscr{G}_{s,n}$. Further, the number of connected components in $\mathscr{H}$ is $s^{n-l}$.
\end{lemma}
\begin{proof}
Recall that the symbol matrix $P$ has a block matrix form, given in equation \eqref{P}. Taking the transpose of each $P_i$, $i\in [n-l+1:n]$ thus gives us:
\begin{equation}\label{PT}
    {[P_i^{T}]}_{\mathbf{r},\mathbf{t}} = \begin{cases}
    1,\text{ if }\textbf{t}_{\setminus i} = \textbf{r}_{\setminus i},\\
    0,\text{ otherwise}
    \end{cases}
\end{equation}
where $\mathbf{r}\in \mathscr{S}^{n}$ and $\mathbf{t}\in \mathscr{S}_{i\leftarrow *}^{n}$.

Since $i$ ranges from $n-l+1$ to $n$ only, by comparing the equation above with \eqref{V}, we get that $P^{T}$ is a submatrix of $V_{\mathscr{G}_{s,n}}$, containing only the first $ls^{n-1}$ columns of $V_{\mathscr{G}_{s,n}}$. We denote by $\mathscr{H}$, the subgraph induced by this submatrix; thus, $\mathscr{H}$ is a sub-hypergraph of $\mathscr{G}_{s,n}$.

Now, let $t$ be the number of connected components in $\mathscr{H}$ and let $h_i$, $1\leq i\leq t$, be the number of hyperedges in each component $\mathscr{H}_i$. Then,
\begin{equation}\label{conn}
    \sum_{i=1}^{t}h_i = l\beta = ls^{n-1}.
\end{equation}
Pick some node indexed by vector $\mathbf{v}$ so that $\mathbf{v}$ belongs to connected component $\mathscr{H}_j$. Consider the collection of nodes $\mathscr{W}$, where $\mathscr{W} = \{\mathbf{w}\in \mathscr{S}^{n}:\text{ }w_i=v_i,\text{ }i\in [n-l]\text{ and } w_j\in [0:s-1], j\in [n-l+1:n]\}$. Note that the set $\mathscr{W}$ includes the node $\mathbf{v}$. From equation \eqref{P_1}, it is easy to verify that the sub-hypergraph of $\mathscr{H}$ induced by the nodes in $\mathscr{W}$ forms the connected component $\mathscr{H}_j$. This can be seen by choosing some node $\mathbf{x} \not\in \mathscr{W}$. Node $\mathbf{x}$ differs from any node in $\mathscr{W}$ in at least one position $j\in [n-l]$. Hence, the row corresponding to $\mathbf{x}$ in $P^{T}$ will have as support, those columns where none of the nodes in $\mathscr{W}$ have a 1 entry, thereby implying that there does not exist a path from $\mathbf{x}$ to any of the nodes in $\mathscr{W}$. Observe that the number of nodes in $\mathscr{H}_j$ is $s^{l}$. 

Since the coordinates of nodes $\mathbf{w}$ in $\mathscr{H}_j$ in positions $i\in [n-l]$ are fixed, we puncture the vectors $\mathbf{w}$ at these positions, to form the vectors $\mathbf{w}'$. Thus, the sets $\mathscr{S}^{n}_{i\leftarrow *}$, $i\in [n-l+1:n]$ can be written as the sets $\mathscr{S}^{l}_{i\leftarrow *}$, $i\in [l]$, each set now containing vectors $\mathbf{w}'$. 

The incidence matrix $V_{\mathscr{H}_j}$ has entries
\begin{equation}\label{vhj}
    {[V_{\mathscr{H}_j}]}_{\mathbf{r},\mathbf{t}} = \begin{cases}
    1,\text{ if }\textbf{t}_{\setminus i} = \textbf{r}_{\setminus i},\\
    0,\text{ otherwise}
    \end{cases}
\end{equation}
where $\mathbf{r}\in \mathscr{S}^{l}$ and $\mathbf{t} \in \mathscr{S}^{l}_{i\leftarrow *}$, $i\in [l]$. 

$\mathscr{H}_j$ is precisely the hypergraph $\mathscr{G}_{s,l}$, having $h_j = ls^{l-1}$ edges. Since this is true for any $j\in [t]$, substituting in equation \eqref{conn}, we get that the number of connected components in $\mathscr{H}$ equals $s^{n-l}$.
\end{proof}

We shall now introduce an $(s-1)^{l}$-dimensional code $\mathscr{C}^{\otimes l}$, of block length $s^{l}$, the parity check matrix of which will aid us in obtaining a handle on the rank of $V_{\mathscr{G}_{s,l}}$.

Consider the single parity check code $\mathscr{C}_{s}$ of block length $s$ over the field $\mathbb{F}$, having the $(s-1)\times s$ generator matrix $G_{s}$ given by
\begin{equation*}
    G_{s} = \begin{bmatrix}
        1 & 0 & \cdots & 0 & -1\\
        0 & 1 & \cdots & 0 & -1\\
        \vdots & \vdots &\ddots &\vdots &\vdots\\
        0 & 0 & \cdots & 1 & -1
    \end{bmatrix}.
\end{equation*}

The parity check matrix, $H_{s}$, of $\mathscr{C}_{s}$ is simply the $1\times s$ all-ones matrix. We then define the direct product code $\mathscr{C}_{s}^{\otimes l}$ (where $\mathscr{C}_{s}$ is the underlying code), as the code generated by
\begin{align*}
    G_{s}^{\otimes l} &= \underbrace{G_s \otimes G_s \otimes \cdots \otimes G_s}  \\
    & \ \ \ \ \ \ \ \ \ \ \ \ l \text{ times}
\end{align*}
where $\otimes$ denotes the Kronecker product.
A codeword in $\mathscr{C}_{s}^{\otimes l}$ is of length $s^l$ and can be described by an $l$-dimensional array. Each entry of the array (which is a coordinate of the codeword) is indexed by an $s$-ary $l$-tuple $\mathbf{v} = (v_1,v_2,\ldots,v_l)\in \mathscr{S}^{l}$. 

The code $\mathscr{C}_{s}^{\otimes l}$ has the property that each array element $\mathbf{v}$ is involved in $l$ parity check equations, one along each coordinate $i\in [l]$. In other words, for every symbol $\mathbf{v}\in \mathscr{S}^{l}$, there exists a parity check equation indexed by a vector $\mathbf{b}\in S^{l}_{i\leftarrow *}$, $i\in [l]$, such that $\textbf{b}_{\setminus i} = \textbf{v}_{\setminus i}$. Moreover, the parity check equation along coordinate $j$ is the sum of those codeword symbols that differ from $\mathbf{v}$ in only their $j^{th}$ coordinate.

Formally, the code $\mathscr{C}_{s}^{\otimes l}$ can be described by the $ls^{l-1} \times s^l$ parity check matrix $H$ having entries
\begin{equation}\label{H}
    H_{\mathbf{r},\mathbf{t}} = \begin{cases}
    1,\text{ if }\textbf{t}_{\setminus i} = \textbf{r}_{\setminus i},\\
    0,\text{ otherwise}
    \end{cases}
\end{equation}
where $\mathbf{r}\in S^{l}_{i\leftarrow *}$, $i\in [l]$ and $\mathbf{t}\in \mathscr{S}^{l}$. Thus, each row of $H$ corresponds to a parity check equation $\mathbf{b}$, that finds the sum of symbols $\mathbf{v}\in \mathscr{S}^{l}$, which differ only in that coordinate of the $l$-tuple, $i$, in which $b_i=*$.

\begin{lemma}\label{kro}
The parity check matrix $H$ of $\mathscr{C}_{s}^{\otimes l}$ is equal to $V_{\mathscr{G}_{s,l}}^{T}$. Further, rank$(V_{\mathscr{G}_{s,l}}^{T})=s^{l}-(s-1)^{l}$.
\end{lemma}

\begin{proof}
The first part of the lemma is obvious from equations \eqref{V} and \eqref{H}. Further, rank$(G_{s}^{\otimes l}) = \prod_{i=1}^{l}\text{rank}(G_s) = \text{rank}(G_{s})^{l} = (s-1)^{l}$. Thus, the rank of the parity check matrix $H$ is equal to $s^{l}-(s-1)^{l}$.
\end{proof}

Using Lemmas \ref{extra} and \ref{kro}, we shall now prove Theorem \ref{s}.

\begin{proof}
Recall from Lemma \ref{extra} that for any node $\mathbf{v}\in \mathscr{S}^{n}$ of the hypergraph $\mathscr{H}$, the connected component $\mathscr{H}_j$ containing $\mathbf{v}$ consists of nodes in the set $\mathscr{W} = \{\mathbf{w}\in \mathscr{S}^{n}:\text{ }w_i=v_i,\text{ }i\in [n-l]\text{ and } w_j\in [0:s-1], j\in [n-l+1:n]\}$. 

It is now possible to permute the rows $\mathbf{v}$ of $P^{T}$ in lexicographic order of $\mathbf{v}'=(v_{n-l},\ldots,v_1)$ so that all the rows $\mathbf{v}\in \mathscr{S}^{n}$ corresponding to a fixed value of $\mathbf{v}'$ occur together. Thus, the first $s^l$ rows of $P^{T}$ are indexed by vectors $\mathbf{v}$ such that $\mathbf{v}' = (0,0,\ldots,0)$, the next $s^{l}$ rows are indexed by vectors $\mathbf{v}$ with $\mathbf{v}' = (0,0,\ldots,1)$ and so on. Each collection of $s^{l}$ rows corresponding to a particular value of $\mathbf{v} = (v_{n-l},\ldots,v_1)$ forms the incidence matrix of a connected component.

From \eqref{vhj}, we observe that the supports of the rows corresponding to any two connected components $\mathscr{H}_i$ and $\mathscr{H}_j$, $i\neq j$, are disjoint. Hence, the rank of $P^{T}$ is equal to the sum of the ranks of the incidence matrices of the connected components of hypergraph $\mathscr{H}$, induced by $P^{T}$. Since each connected component is precisely the hypergraph $\mathscr{G}_{s,l}$ (from the proof of Lemma \ref{extra}), we get that
\begin{equation*}
    \text{rank}(P^{T}) = s^{n-l}(\text{rank}(V_{\mathscr{G}_{s,l}})) = s^{n-l}(s^{l}-(s-1)^{l}),
\end{equation*}
where the first equality follows from Lemma \ref{extra} and the second from Lemma \ref{kro}.
\end{proof}

\section*{Acknowledgements} N.\ Kashyap's work was supported by a Swarnajayanti Fellowship awarded by the Department of Science and Technology, Government of India.

\bibliographystyle{plain}
{\footnotesize
\bibliography{references}}

\begin{thebibliography}{}

\bibitem[\protect\citename{Dimakis {\em et~al.\ }\relax, }2010]{Dimetal}
Dimakis, A.~G., Godfrey, P.~B., Wu, Y., Wainwright, M.~J., \& Ramchandran, K.
  2010.
\newblock Network Coding for Distributed Storage Systems.
\newblock {\em IEEE Transactions on Information Theory}, {\bf 56}(9),
  4539--4551.

\bibitem[\protect\citename{Goparaju {\em et~al.\ }\relax, }2013]{Gopsecure}
Goparaju, S., Rouayheb, S.~El, Calderbank, A.~R., \& Poor, H.~V. 2013.
\newblock Data Secrecy in Distributed Storage Systems under Exact Repair.
\newblock {\em CoRR}, {\bf abs/1304.3156}.

\bibitem[\protect\citename{Goparaju {\em et~al.\ }\relax, }2017]{allparam}
Goparaju, S., Fazeli, A., \& Vardy, A. 2017.
\newblock Minimum Storage Regenerating Codes for All Parameters.
\newblock {\em IEEE Transactions on Information Theory}, {\bf 63}(10),
  6318--6328.

\bibitem[\protect\citename{Huang {\em et~al.\ }\relax, }2017]{Huang}
Huang, K., Parampalli, U., \& Xian, M. 2017.
\newblock On Secrecy Capacity of Minimum Storage Regenerating Codes.
\newblock {\em IEEE Transactions on Information Theory}, {\bf 63}(3),
  1510--1524.

\bibitem[\protect\citename{Pawar {\em et~al.\ }\relax, }2010]{Upper}
Pawar, S., Rouayheb, S.~El, \& Ramchandran, K. 2010 (June).
\newblock On secure distributed data storage under repair dynamics.
\newblock {\em Pages  2543--2547 of:} {\em 2010 IEEE International Symposium on
  Information Theory}.

\bibitem[\protect\citename{Rashmi {\em et~al.\ }\relax, }2011]{Rash2}
Rashmi, K.~V., Shah, N.~B., \& Kumar, P.~V. 2011.
\newblock Optimal Exact-Regenerating Codes for Distributed Storage at the MSR
  and MBR Points via a Product-Matrix Construction.
\newblock {\em IEEE Transactions on Information Theory}, {\bf 57}(8),
  5227--5239.

\bibitem[\protect\citename{Rawat, }2017]{Rawat}
Rawat, A.~S. 2017 (June).
\newblock Secrecy capacity of minimum storage regenerating codes.
\newblock {\em Pages  1406--1410 of:} {\em 2017 IEEE International Symposium on
  Information Theory (ISIT)}.

\bibitem[\protect\citename{Rawat {\em et~al.\ }\relax, }2014]{vish}
Rawat, A.~S., Koyluoglu, O.~O., Silberstein, N., \& Vishwanath, S. 2014.
\newblock Optimal Locally Repairable and Secure Codes for Distributed Storage
  Systems.
\newblock {\em IEEE Transactions on Information Theory}, {\bf 60}(1), 212--236.

\bibitem[\protect\citename{Shah {\em et~al.\ }\relax, }2011]{Shah}
Shah, N.~B., Rashmi, K.~V., \& Kumar, P.~V. 2011 (Dec).
\newblock Information-Theoretically Secure Regenerating Codes for Distributed
  Storage.
\newblock {\em Pages  1--5 of:} {\em 2011 IEEE Global Telecommunications
  Conference - GLOBECOM 2011}.

\bibitem[\protect\citename{Shah {\em et~al.\ }\relax, }2012]{interior}
Shah, N.~B., Rashmi, K.~V., Kumar, P.~V., \& Ramchandran, K. 2012.
\newblock Distributed Storage Codes With Repair-by-Transfer and
  Nonachievability of Interior Points on the Storage-Bandwidth Tradeoff.
\newblock {\em IEEE Transactions on Information Theory}, {\bf 58}(3),
  1837--1852.

\bibitem[\protect\citename{Ye \& Barg, }2017]{ye}
Ye, M., \& Barg, A. 2017.
\newblock Explicit Constructions of High-Rate MDS Array Codes With Optimal
  Repair Bandwidth.
\newblock {\em IEEE Transactions on Information Theory}, {\bf 63}(4),
  2001--2014.

\end{thebibliography}
\end{document}